\documentclass[conference]{IEEEtran}
\usepackage{times}

\usepackage[numbers]{natbib}
\usepackage{multicol}
\usepackage[bookmarks=true]{hyperref}
\usepackage{hyperref}
\usepackage{graphicx}		
\usepackage{wrapfig}
\usepackage[format=plain,font=footnotesize,labelfont=bf,labelsep=period]{caption}
\usepackage[export]{adjustbox}
\usepackage{bm}

\usepackage{amsmath}
\usepackage{mathrsfs}
\usepackage{amssymb}
\usepackage{mathtools}
\usepackage{algorithm}
\usepackage{algpseudocode}

\usepackage{pdfsync}
\usepackage[normalem]{ulem}
\usepackage{paralist}	
\usepackage[space]{grffile} 

\usepackage{color}

\usepackage{amsthm}
\usepackage{centernot}

\newtheorem{theorem}{Theorem}

\newtheorem{lemma}{Lemma}
\newtheorem*{lemma*}{Lemma}
\newtheorem{definition}{Definition}
\newtheorem{remark}{Remark}

\newtheorem*{proposition*}{Proposition}

\newcommand{\R}{\mathbb{R}}

\newcommand{\C}{\mathcal{C}}

\definecolor{darkblue}{RGB}{0,0,102}
\definecolor{lightblue}{RGB}{77,77,148}

\definecolor{gold}{RGB}{234, 170, 0}
\definecolor{metallic_gold}{RGB}{139, 111, 78}

\newcommand{\mb}[1]{\mathbf{ #1 }}
\newcommand{\bs}[1]{\boldsymbol{ #1 }}

\DeclareMathOperator*{\argmin}{argmin}

\DeclareMathOperator{\diag}{diag}

\newcommand{\lmat }{\begin{bmatrix}}
\newcommand{\rmat}{\end{bmatrix}}

\renewcommand{\P}{\mathbb{P}} 
\newcommand{\E}{\mathbb{E}} 
\usepackage{dsfont}

\usepackage{svg}

\begin{document}

\title{Robust Safety under Stochastic Uncertainty with Discrete-Time Control Barrier Functions}


\author{\authorblockN{Ryan K. Cosner$^1$, Preston Culbertson$^1$, Andrew J. Taylor$^2$, and Aaron D. Ames$^1$}
\vspace{0.7em}
\authorblockA{
$^1$Department of Mechanical and Civil Engineering, 
$^2$Computing and Mathematical Sciences Department\\
California Institute of Technology, 
Pasadena, California 91125\\
\{rkcosner, pculbert, ajtaylor, ames\}@caltech.edu}
\vspace{-2.5em}
}



%

\maketitle

\begin{abstract}
Robots deployed in unstructured, real-world environments operate under considerable uncertainty due to imperfect state estimates, model error, and disturbances. Given this real-world context, the goal of this paper is to develop controllers that are provably safe under uncertainties.  To this end, we leverage Control Barrier Functions (CBFs) which guarantee that a robot remains in a ``safe set'' during its operation---yet CBFs (and their associated guarantees) are traditionally studied in the context of continuous-time, deterministic systems with bounded uncertainties.    
In this work, we study the safety properties of discrete-time CBFs (DTCBFs) for systems with discrete-time dynamics and unbounded stochastic disturbances. Using tools from martingale theory, we  develop probabilistic bounds for the safety (over a finite time horizon) of systems whose dynamics satisfy the discrete-time barrier function condition in expectation, and analyze the effect of Jensen's inequality on DTCBF-based controllers. Finally, we present several examples of our method synthesizing safe control inputs for systems subject to significant process noise, including an inverted pendulum, a double integrator, and a quadruped locomoting on a narrow path. 

\end{abstract}
\IEEEpeerreviewmaketitle

\section{Introduction}

Safety is critical for a multitude of modern robotic systems: from autonomous vehicles, to medical and assistive robots, to aerospace systems. When deployed in the real world, these systems face sources of uncertainty such as imperfect perception, approximate models of the world and the system, and unexpected disturbances. In order to achieve the high degrees of safety necessary for these robots to be deployed at scale, it is essential that controllers can not only guarantee safe behavior, but also provide robustness to these uncertainties. 

In the field of control theory, safety is often defined as the forward invariance of a ``safe set'' \cite{ames2016control}. In this view, a closed-loop system is considered safe if all trajectories starting inside the safe set will remain in this set for all time. Several tools exist for generating controllers which can guarantee this forward-invariance property, including Control Barrier Functions (CBFs) \cite{ames_control_2019}, reachability-based controllers \cite{bansal2017hamilton}, and state-constrained Model-Predictive Controller (MPC) approaches \cite{hewing2020learning}. 
Considerable advancements have been made in guaranteeing safety or stability in the presence of bounded uncertainties \cite{zhou1998essentials, blanchini2008set, aubin2011viability, sontag2008input,kolathaya2018input,alan2021safe}.  Yet less attention has been paid to the case of unbounded uncertainties, where the aforementioned methods generally do not apply. 

Obtaining robust safety in the case of unbounded disturbances is particularly important when considering systems subject to stochastic disturbances, since these disturbances are often modeled as continuous random variables with unbounded support (e.g., zero-mean, additive Gaussian noise); for such systems, it is impossible to give an absolute bound on the disturbance magnitude. Existing methods for unbounded, random disturbances fall into two categories. The first is to impose step-wise chance constraints on a given safety criterion (e.g., a state constraint in MPC \cite{hewing2020learning} or CBF-based controllers \cite{ahmadi_risk-averse_2022}), which in turn provide one-step safety guarantees. The other class of approaches \cite{kushner1967stochastic, prajna2004stochastic, santoyo_verification_2019, clark_control_2019, steinhardt2012finite} use Lyapunov or barrier function techniques to provide bounds on the safety probabilities for trajectories over a fixed time horizon; existing approaches, however, often assume the presence of a stabilizing controller, or model the system in continuous-time (i.e., assume the controller has, in effect, infinite bandwidth).

\begin{figure}[t]
    \centering
    \includegraphics[width=\linewidth]{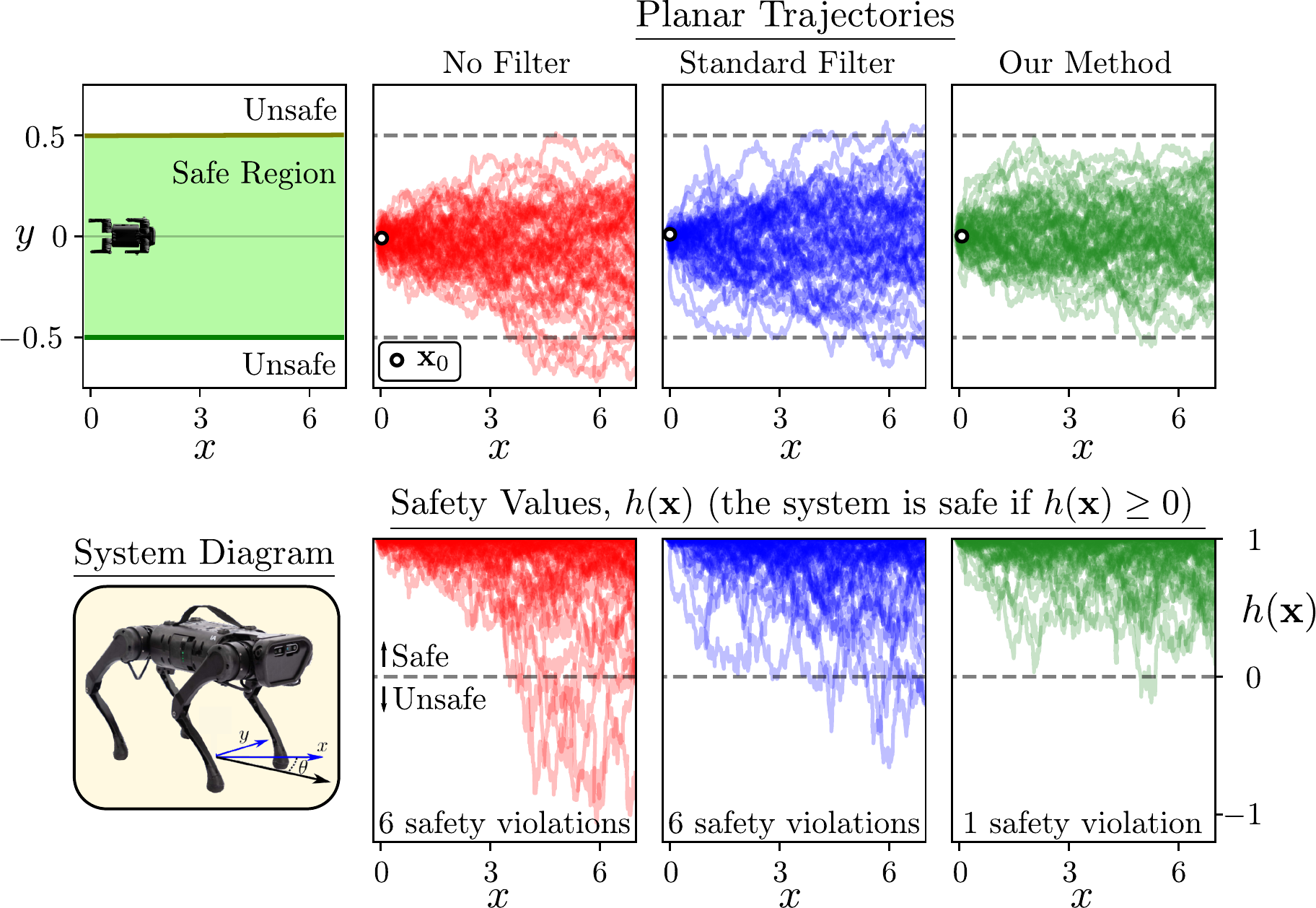}
    \caption{Safety of a simulated quadrupedal robot locomoting on a narrow path for a variety of controllers. \textbf{(Top Left)} The safe region that the quadruped is allowed to traverse. \textbf{(Bottom Left)} A system diagram depicting the states of the quadruped $\lmat x, y, \theta \rmat^\top$. \textbf{(Top Right)} 50 trajectories for 3 controllers: one without any knowledge of safety  ($\mb{k}_{\textrm{nom}}$), one with a standard safety filter \eqref{eq:dtcbfop}, and finally our method which accounts for stochasticity \eqref{eq:jed}. \textbf{(Bottom Right)} Plots of $h(\mb{x})$, a scalar value representing safety. The system is safe (i.e., in the green safe region) if $h(\mb{x}) \geq 0 $. }
    \label{fig:hero_fig}
        \vspace{-1.1cm}
\end{figure}

In order to best represent the uncertainty that might appear from sources such as discrete-time perception errors or sampled-data modeling errors, we focus our work on generating probabilistic bounds of safety for discrete-time (DT) stochastic systems. While MPC state constraints are generally enforced in discrete time, CBFs, normally applied in continuous time, have a discrete-time counterpart (DTCBFs) that were first introduced in \cite{agrawal2017discrete} and have gained popularity due to their compatibility with planners based on MPC \cite{zeng2021safety,liu2022iterative,wills2004barrier}, reinforcement learning \cite{cheng2019end}, and  Markov decision processes \cite{ahmadi2019safe}. 
In a stochastic setting, martingale-based techniques have been leveraged to establish safety guarantees \cite{santoyo_verification_2019,steinhardt2012finite}, yet these works have limited utility when analyzing the safety of discrete-time CBF-based controllers. 

In particular, the ``c-martingale'' condition used in \cite{steinhardt2012finite} does not admit a multiplicative scaling of the barrier function, and therefore, at best, provides a weak worst-case safety bound for CBF-based controllers that grows linearly in time.  The work of \cite{santoyo_verification_2019} (which builds upon \cite{kushner1967stochastic}, as does this paper) is largely focused on offline control synthesis to achieve a desired safety bound (as opposed to the online, optimization-based control studied in this work). Also, the method proposed in \cite{santoyo_verification_2019} can only generate discrete-time controllers for affine barriers, which severely limits its applicability to general barrier functions. Both papers also depend on  sum-of-squares (SoS) programming \cite{papachristodoulou2005tutorial} for control synthesis/system verification, thereby requiring an offline step that scales poorly with the state dimension.  The goal of this paper is to extend the results of \cite{kushner1967stochastic} in a different direction, and thereby enable the synthesis of online controllers that can be realized on robotic systems.


The main contribution of this paper is to apply martingale-based probability bounds in the context of discrete-time CBFs to guarantee robust safety under stochastic uncertainty.  To this end, we leverage the bounds originally presented in the seminal work by Kushner \cite{kushner1967stochastic}.  
Our first key contribution is the translation of these results from a Lyapunov setting to a CBF one.  To this end, we present a new proof of the results in \cite{kushner1967stochastic} which we believe to be more complete and intuitive and which relates to the existing results of Input-to-State Safety (ISSf) for systems with bounded uncertainties \cite{kolathaya2018input}. 
Furthermore, we present a method (based on Jensen's inequality) to account for the effects of process noise on a DTCBF-based controller. Finally, we apply this method to a variety of systems in simulation to analyze the tightness of our bound and demonstrate its utility. These experiments range from simple examples that illustrate the core mathematics---a single and double integrator and a pendulum---to a high fidelity simulation of a quadrupedal robot locomoting along a narrow path with the uncertainty representing the gap between the simplified and full-order dynamics models.

\section{Background}

In this section we provide a review of safety for discrete-time nonlinear systems via  control barrier functions (CBFs), and review tools from probability theory useful for studying systems with stochastic disturbances.

\subsection{Safety of Discrete-time Systems}

Consider a discrete-time (DT) nonlinear system with dynamics given by: 
\begin{align}
    \mb{x}_{k+1} = \mb{F}(\mb{x}_k, \mb{u}_k), \quad \forall k \in \mathbb{N},  \label{eq:dt_dyn}
\end{align}
with state $\mb{x}_k \in \R^n$, input $\mb{u}_k \in \R^m$, and continuous dynamics $\mb{F}: \R^{n} \times \R^{m} \to \R^n$. A continuous state-feedback controller $\mb{k}:\R^n\to\R^m$ yields the DT closed-loop system: 
\begin{align}
    \mb{x}_{k+1} = \mb{F}(\mb{x}_k, \mb{k}(\mb{x}_k)), \quad \forall k \in \mathbb{N}. \label{eq:dt_autonomous}
\end{align}

We formalize the notion of safety for systems of this form using the concept of forward invariance: 

\begin{definition}[Forward Invariance \& Safety \cite{blanchini2008set}]
A set $\mathcal{C}\subset \R^n $ is \textit{forward invariant} for the system \eqref{eq:dt_autonomous} if $\mb{x}_0 \in \mathcal{C}$ implies that $\mb{x}_k \in \mathcal{C}$ for all $k \in \mathbb{N}$. In this case, we call the system \eqref{eq:dt_autonomous} \textit{safe} with respect to the set $\mathcal{C}$. 
\end{definition}

Discrete-time barrier functions (DTBFs) are a tool for guaranteeing the safety of discrete-time systems. Consider a set $\mathcal{C} \triangleq \left \{ \mb{x} \in \R^n \mid h(\mb{x}) \geq 0 \right \}$ expressed as the 0-superlevel set of a continuous function $h:\R^n\to\R$. We refer to such a function $h$ as a DTBF\footnote{ The state constraint $\mb{x}_k \in \mathcal{C}$, when expressed as $h(\mb{x}_k) \geq 0 $, is the special case of a DTBF with $\alpha = 0 $. } if it satisfies the following properties:

\begin{definition}[Discrete-Time Barrier Function (DTBF) \cite{agrawal2017discrete}]
    Let $\mathcal{C}\subset \R^n $ be the $0$-superlevel set of a continuous function $h:\R^n \to \R$. The function $h$ is a \textit{discrete-time barrier function} (DTBF) for \eqref{eq:dt_autonomous} on $\mathcal{C}$ if there exists an $\alpha \in [0, 1] $ such that for all $\mb{x} \in \R^n$, we have that: 
    \begin{align}
        h(\mb{F}(\mb{x}, \mb{k}(\mb{x})))   \geq \alpha h(\mb{x}). \label{eq:dtbf_constraint}  
    \end{align}
\end{definition}

\noindent This inequality mimics that of discrete-time Lyapunov functions \cite{bof2018lyapunov}, and similarly regulates the evolution of $h$ based on its previous value. DTBFs serve as a certificate of forward invariance as captured in the following theorem: 

\begin{theorem}[\cite{agrawal2017discrete}]
    Let $\mathcal{C}\subset \R^n $ be the $0$-superlevel set of a continuous function $h:\R^n \to \R$. If $h$ is a DTBF for \eqref{eq:dt_autonomous} on $\mathcal{C}$, then the system \eqref{eq:dt_autonomous} is safe with respect to the set $\mathcal{C}$. 
\end{theorem}

\noindent Intuitively, the value of $h(\mb{x}_k)$ can only decay as fast as the geometric sequence $\alpha^kh(\mb{x}_0)$, which is lower-bounded by 0, thus ensuring the safety (i.e., forward invariance) of $\C$.


Discrete-time control barrier functions (DTCBFs) provide a tool for constructively synthesizing controllers that yield closed-loop systems that possess a DTBF:

\begin{definition}[Discrete-Time Control Barrier Function (DTCBF) \cite{agrawal2017discrete}]
    Let $\mathcal{C}\subset \R^n $ be the $0$-superlevel set of a continuous function $h:\R^n \to \R$. The function $h$ is a \textit{discrete-time control barrier function} (DTCBF) for \eqref{eq:dt_dyn} on $\mathcal{C}$ if there exists an $\alpha \in [0, 1] $ such that for each $\mb{x} \in \R^n$, there exists a $\mb{u}\in\R^m$ such that: 
    \begin{align}
        h(\mb{F}(\mb{x}, \mb{u}))   \geq \alpha h(\mb{x}). \label{eq:dtcbf_constraint}  
    \end{align}
\end{definition}

Given a CBF $h$ for \eqref{eq:dt_dyn} and a corresponding $\alpha\in[0,1]$, we define the point-wise set of control values:
\begin{equation}
    \mathscr{K}_{\rm CBF}(\mb{x}) = \left\{ \mb{u}\in\R^m \mid h(\mb{F}(\mb{x},\mb{u})) \geq \alpha h(\mb{x})\right\}.
\end{equation}
This yields the following result:

\begin{theorem}[\cite{agrawal_constructive_2022}] \label{thm:dtcbf}
    Let $\mathcal{C}\subset \R^n $ be the $0$-superlevel set of a continuous function $h:\R^n \to \R$. If $h$ is a DTCBF for \eqref{eq:dt_dyn} on $\mathcal{C}$, then the set $\mathscr{K}_{\rm CBF}(\mb{x})$ is non-empty for all $\mb{x}\in\R^n$, and for any continuous state-feedback controller $\mb{k}$ with $\mb{k}(\mb{x})\in \mathscr{K}_{\rm CBF}(\mb{x})$ for all $\mb{x}\in\R^n$, the function $h$ is a DTBF for \eqref{eq:dt_autonomous} on $\C$.
\end{theorem}

Given a continuous nominal controller $\mb{k}_{\rm nom}:\R^n\times \mathbb{N}\to\R^m$ and a DTCBF $h$ for \eqref{eq:dt_dyn} on $\C$, a controller $\mb{k}$ satisfying $\mb{k}(\mb{x},k)\in \mathscr{K}_{\rm CBF}(\mb{x})$ for all $\mb{x}\in\R^n$ and $k \in \mathbb{N}$ can be specified via the following optimization problem:
\begin{align}
\label{eq:dtcbfop}
\mb{k}(\mb{x}) = \argmin_{\mb{u}\in\R^m}&\quad \Vert \mb{u}-\mb{k}_{\rm nom}(\mb{x},k) \Vert^2 \tag{DTCBF-OP}\\ \textrm{s.t.}& \quad h(\mb{F}(\mb{x},\mb{u})) \geq \alpha h(\mb{x}).\nonumber
\end{align}

We note that unlike the affine inequality constraint that arises with continuous-time CBFs \cite{ames_control_2019}, the DTCBF inequality constraint \eqref{eq:dtcbf_constraint} is not necessarily convex with respect to the input, preventing it from being integrated into a convex optimization-based controller. To solve this issue, it is often assumed that the function $h\circ\mb{F}:\R^n\times\R^m\to\R$ is concave with respect to its second argument \cite{agrawal2017discrete, ahmadi2019safe, zeng2021safety}. This assumption was shown to be well motivated for concave $h$ \cite{taylor_safety_2022}.

\subsection{Stochastic Preliminaries}

We now review tools from probability theory that will allow us to utilize information about the distribution of a stochastic disturbance signal in constructing a notion of stochastic safety and corresponding safety-critical controllers. We choose to provide this background material at the level necessary to understand our later constructions of stochastic safety and safety-critical controllers, but refer readers to \cite{grimmett2020probability} for a precise measure-theoretic presentation of the following concepts.

The key tool underlying our construction of a notion of stochastic safety is a nonnegative supermartingale, a specific type of expectation-governed random process: 

\begin{definition}
Let $\mb{x}_k$ be a sequence of random variables that take values in $\R^n$, $W:\R^n\times\mathbb{N}\to\R$, and suppose that $\mathbb{E}\big[\lvert W(\mb{x}_k,k) \rvert\big] <\infty$ for $k\in \mathbb{N}$. The process $W_k\triangleq W(\mb{x}_k,k)$ is a supermartingale if:
\begin{equation}
    \label{eq:supermartingale}
        \mathbb{E}[ W_{k+1} \mid \mb{x}_{0:k}] \leq W_k~\textrm{almost~surely~for~all~}k\in\mathbb{N},
\end{equation}
    where $\mb{x}_{0:k}$ indicates the random variables $\left\{\mb{x}_0, \mb{x}_1, \ldots, \mb{x}_k\right\}$. If, additionally, $W_k\geq 0$ for all $k\in\mathbb{N}$, $W_k$ is a nonnegative supermartingale. If the process is non-decreasing in expectation, the process $W_k$ is a submartingale. If the inequality \eqref{eq:supermartingale} holds with equality, the process $W_k$ is a martingale. 
\end{definition}

An important result from martingale theory that we will use to develop probabilistic safety guarantees is \textit{Ville's inequality}, which allows us to bound the probability that a nonnegative supermartingale will rise above a certain value: 

\begin{theorem}[Ville's Inequality \cite{ville1939etude}]
    Let $W_k$ be a nonnegative supermartingale. Then for all $\lambda\in\R_{>0}$, 
    \begin{align}
        \P \left\{ \sup_{k\in \mathbb{N}} W_k > \lambda  \right\} \leq \frac{\mathbb{E}[W_0]}{\lambda}.
        \label{eq:ville}
    \end{align}
\end{theorem}
Intuitively, Ville's inequality can be compared with Markov's inequality for nonnegative random variables; since the process $W_k$ is nonincreasing in expectation, Ville's inequality allows us to control the probability the process instead moves upward above $\lambda$.

Lastly, as we will see when synthesizing safety-critical controllers in the presence of stochastic disturbances, we will need to enforce conditions on the expectation of a DCTBF. In doing so, we will need to relate the expectation of the DCTBF $h(\mb{x}_{k+1})$ to the expectation of the state $\mb{x}_{k+1}$. This will be achieved using Jensen's inequality:

\begin{theorem}[Jensen's Inequality \cite{liao2018sharpening}]
\label{thm:jensen}
    Consider a continuous function $h: \R^n \to \R$ and a random variable $\mb{x}$ that takes values in $\R^n$ with $\E[\Vert\mb{x}\Vert] < \infty$. We have that: 
    \begin{align}
        \begin{cases}
        \textrm{if $h$ is convex, }& \textrm{then }  \E[h(\mb{x})] \geq h(\E[\mb{x}]) ,\\
        \textrm{if $h$ is concave, } &\textrm{then }  \E[h(\mb{x})] \leq h(\E[\mb{x}]).  
        \end{cases} \label{prop:jensens}
    \end{align}
\end{theorem}


\section{Safety of Discrete-Time Stochastic Systems}\label{sec:main_thm}

In this section we provide one of our main results in the form of a bound on the probability that a system with stochastic disturbances will exit a given superlevel set of a DTBF over a finite time horizon.

Consider the following modification of the DT system \eqref{eq:dt_dyn}:
\begin{align}
    \mb{x}_{k+1} = \mb{F}(\mb{x}_k, \mb{u}_k) + \mathbf{d}_k, \quad \forall k \in \mathbb{N},  \label{eq:dt_dyn_dist}
\end{align}
with $\mb{d}_k$ taking values in $\R^n$, and a closed-loop system: 
\begin{align}
    \mb{x}_{k+1} = \mb{F}(\mb{x}_k, \mb{k}(\mb{x}_k)) + \mb{d}_k, \quad \forall k \in \mathbb{N}.\label{eq:dt_autonomous_dist}
\end{align}
We assume that $\mb{x}_0$ is known and the disturbances $\mb{d}_k$ are a sequence of independent and identically distributed (with distribution $\mathcal{D}$) random variables\footnote{This implies the dynamics define a Markov process, i.e. $\E[h(\mb{F}(\mb{x}_k,\mb{u}_k) + \mb{d}_k)\mid\mb{x}_{0:k}] = \E[h(\mb{F}(\mb{x}_k,\mb{u}_k) + \mb{d}_k)\mid\mb{x}_k],$ since the state $\mb{x}_{k+1}$ at time $k+1$ only depends on the state $\mb{x}_k$, input $\mb{u}_k$, and disturbance $\mb{d}_k$ at time $k$.} with (potentially unbounded) support on $\R^n$, generating the random process $\mb{x}_{1:k}$. To study the safety of this system, we will use the following definition:

\begin{definition}[$K$-Step Exit Probability]
    Let $h:\R^n\to\R$ be a continuous function. For any $K\in\mathbb{N}$, $\gamma\in\R_{\geq 0}$, and initial condition $\mb{x}_0\in\R^n$, the $K$-step exit probability of the closed-loop system \eqref{eq:dt_autonomous_dist} is given by:
    \begin{align}
        P_u(K,\gamma,\mb{x}_0) =  \mathbb{P}\left\{ \min_{k \in \{0, \dots, K\}} h(\mb{x}_k) < - \gamma \right\}.
    \end{align}
\end{definition}
\noindent which describes the probability that the system will leave the $-\gamma$ superlevel set of $h$ within $K$ steps. This probability is directly related to the robust safety concept of Input-to-State Safety (ISSf)  \cite{kolathaya2018input} which reasons about the superlevel set of $h$ which is rendered safe in the presence of bounded disturbances. For the remainder of this work, we will omit the dependence of $P_u$ on $K$, $\gamma$, and $\mb{x}_0$ for notational simplicity.

\begin{remark}
     \textup{The finite time aspect of $K$-step exit probabilities is critical since systems exposed to unbounded disturbances 
     will exit a bounded set with probability $P_u = 1$ over an infinite horizon \cite{steinhardt2012finite, chern_safe_2021}. Intuitively, this is because a sufficiently large sample will eventually be drawn from the tail of the distribution that forces the system out in a single step.}
\end{remark}

Given this definition, we now provide one of our main results relating DTBFs to $K$-step exit probabilities. We note that this result is a reframing of the stochastic invariance theorem in \cite{kushner1967stochastic, santoyo_verification_2019}. Our reframing features three key components. First, we develop our results using the standard formulation of DTBFs covered in the background. Second, we produce a probability bound not only for $\C$ (defined as the 0-superlevel set of $h$, such that $\gamma = 0$), but for all non-positive superlevel sets of $h$ ($\gamma \geq 0$), a stochastic variant of ISSf \cite{kolathaya2018input}. Third, we present a complete proof of our result, with the goal of illuminating how to leverage tools from martingale theory to reason about the safety of discrete-time stochastic systems.

\begin{theorem}\label{thm:kushner_main}
   Let $h:\R^n \to \R$ be a continuous, upper-bounded function with upper bound $M\in\R_{>0}$. Suppose there exists an $\alpha\in (0,1)$ and a\footnote{The original presentation of Theorem \ref{thm:kushner_main} in \cite{kushner1967stochastic} considers variable $\delta_k$ for $k \in \{0, \dots, K\}$, which are known \textit{a priori}. In most practical applications, one assumes a lower bound that holds for all $\delta_k$, motivating our use of a constant $\delta$. Moreover, the use of a constant $\delta$ significantly clarifies the proof.} $\delta \leq M(1-\alpha)$ such that the closed-loop system \eqref{eq:dt_autonomous_dist} satisfies: 
    \begin{align}
        \mathbb{E}[~h(\mb{F}(\mb{x}, \mb{k}(\mb{x})) + \mb{d})  \mid \mb{x}~] \geq \alpha h(\mb{x})+ \delta,  \label{eq:kushner_constraint}
    \end{align}
    for all $\mb{x}\in \R^n$, with $\mb{d}\sim\mathcal{D}$. For any $K \in \mathbb{N}$ and $\gamma \in\R_{\geq 0}$, if $\delta < -\gamma(1 - \alpha)$, we have that:
    \begin{align}
    \label{eq:probup}
        P_u & \leq  \left( \frac{M - h(\mb{x}_0 )}{M + \gamma } \right)\alpha^K +  \frac{M (1 - \alpha) - \delta }{M + \gamma}\sum_{i =1}^K\alpha^{i-1}.  
    \end{align}
    Alternatively if $\delta \geq -\gamma (1 - \alpha) $, then: 
    \begin{align}
    \label{eq:problo}
        P_u\leq 1 - \frac{h(\mb{x}_0) + \gamma }{M+\gamma}\left(  \frac{M\alpha +\gamma + \delta}{M+ \gamma} \right)^K. 
    \end{align}
\end{theorem}

\begin{remark}
    \textup{The upper bound $\delta \leq M(1-\alpha)$ is relatively non-restrictive, as not only is $\delta$ typically negative, but it must hold such that, in expectation, $h(\mb{x}_{k+1})$ cannot rise above the upper bound $M$ on $h$. The switching condition between \eqref{eq:probup} and \eqref{eq:problo} of $\delta = \gamma(1-\alpha)$ corresponds to whether, in expectation, the one-step evolution of the system remains in the set $\mathcal{C}_\gamma = \{ \mb{x} \in \R^n \mid h(\mb{x}) \geq - \gamma \}$ when it begins on the boundary of $\C_\gamma$.}
\end{remark}

To make our argument clear at a high level, we begin with a short proof sketch before proceeding in detail.

\textit{Proof sketch: } The key tool in proving Theorem \ref{thm:kushner_main} is Ville's inequality \eqref{eq:ville}. Since $h(\mb{x}_k)$, in general, is not a super- or submartingale, we will first construct a nonnegative supermartingale, $W_k \triangleq W(\mb{x}_k, k)$, by scaling and shifting $h(\mb{x}_k)$. We can then apply Ville's inequality \eqref{eq:ville} to bound the probability of $W_k$ going above any $\lambda > 0$. Next we find a particular value of $\lambda$, denoted $\lambda^*$, such that:
\begin{equation}
\max_{k \in \{0, \ldots, K\}} W_k \leq \lambda^* \implies \min_{k \in \{0, \ldots, K\}} h(\mb{x}_k) \geq -\gamma.
\end{equation}
Intuitively, this means that any sequence $W_k$ that  remains below $\lambda^*$ ensures that the corresponding sequence $h(\mb{x}_k)$ remains (safe) above $-\gamma$. This allows us to bound the $K$-step exit probability $P_u$ of our original process $h(\mb{x}_k)$ with the probability that $W_k$ will rise above $\lambda^*$:
\begin{align}
    P_u \leq \mathbb{P}\left\{\max_{k \in \{0, \ldots, K\}} W_k > \lambda^*\right\} \leq  \frac{\mathbb{E}[W_0]}{\lambda^*} = \frac{W_0}{\lambda^*},  
\end{align}
where the last equality will follow as it is assumed $\mb{x}_0$ is known \textit{a priori}. Particular choices of $W$ and $\lambda^*$ will yield the bounds stated in the theorem, completing the proof.

\subsection{Proof: Constructing a Nonnegative Supermartingale}

We will begin by constructing a nonnegative supermartingale, allowing us to use Ville's inequality. To construct this supermartingale, we first note that by rearranging terms in the inequality in \eqref{eq:kushner_constraint}, we can see the process $M - h(\mb{x}_k)$ resembles a supermartingale:
\begin{align}
  \E[M - h(\mb{x}_{k+1})\mid \mb{x}_k] &\leq \alpha (M - h(\mb{x}_k)) + M(1-\alpha) -\delta,\nonumber\\
  & \triangleq \alpha(M - h(\mb{x}_k)) + \varphi,  \label{eq:phidef}
\end{align}
but with a scaling $\alpha$ and additive term $\varphi \triangleq M(1-\alpha) - \delta$ that makes $\mathbb{E}\left[M - h(\mb{x}_{k+1}) \mid \mb{x}_k\right] \nleq M - h(\mb{x}_k)$ in general. To remove the effects of $\alpha$ and $\varphi$, consider the function $W: \R^n \times \mathbb{N} \to \R$ defined as: 
\begin{align}
    W(\mb{x}_k, k)  & \triangleq \underbrace{(M - h(\mb{x}_k))\theta^k}_\textrm{negate and scale} -\underbrace{\varphi\sum_{i=1}^{k} \theta^{i}}_{\textrm{cancel $\varphi$}} + \underbrace{\varphi \sum_{i=1}^{K} \theta^{i}}_{\textrm{ensure $W \geq 0$}},
    \label{eq:W_expanded}
\end{align}
where $\theta \in [1, \infty)$ will be used to cancel the effect of $\alpha$, but is left as a free variable that we will later use to tighten our bound on $P_u$. Denoting $W_k \triangleq W(\mb{x}_k,k)$, we now verify $W_k$ is a nonnegative supermartingale. We first show that $W_k \geq 0$ for all $k \in \{0, \dots, K\}$. Combining the two sums in \eqref{eq:W_expanded} yields:
\begin{align}
\label{eq:compactWk}
    W_k = (M - h(\mb{x}_k))\theta^k + \varphi \sum_{i=k+1}^K \theta^i,
\end{align}
which is nonnegative as $h(\mb{x}) \leq M$ for all $\mb{x} \in \R^n$, $\theta \geq 1$, and $\varphi \geq 0$ since $\delta \leq M(1-\alpha)$ by assumption. We now show that $W_k$ satisfies the supermartingale inequality \eqref{eq:supermartingale}:
\begin{align}
    \textcolor{black}{\mathbb{E}}&\textcolor{black}{\left[W_{k+1} \mid \mb{x}_{0:k} \right] = \E[W_{k+1} \mid \mb{x}_k ],} \label{eq:markov}\\ &= (M-\mathbb{E}[h(\mb{x}_{k+1})\mid \mb{x}_k])\theta^{k+1} + \varphi \sum_{i=k+2}^K \theta^i, \label{eq:wkp1def}\\
    &\leq (M - \alpha h(\mb{x}_k) - \delta)\theta^{k+1} + \varphi \sum_{i=k+2}^K \theta^i,\label{eq:w_barrier_cond}\\
    &= \alpha \theta (M - h(\mb{x}_k))\theta^k + \theta^{k+1}\underbrace{((1-\alpha) M - \delta)}_{=\varphi} + \varphi \sum_{i=k+2}^K \theta^i,\nonumber\\
    &= \underbrace{\alpha \theta}_{\text{req.} \leq 1} (M - h(\mb{x}_k)) \theta^k+ \varphi \sum_{i=k+1}^K \theta^i \leq W_k \label{eq:w_theta_const},
\end{align}
where \eqref{eq:markov} is due to the Markovian nature of system \eqref{eq:dt_autonomous_dist}, \eqref{eq:wkp1def} comes from using \eqref{eq:compactWk} to write $W_{k+1}$, \eqref{eq:w_barrier_cond} follows from \eqref{eq:kushner_constraint}, and \eqref{eq:w_theta_const} follows from the preceding line using the definition of $\varphi$ and assuming the further requirement that $\theta \leq \frac{1}{\alpha}$. Thus, we have shown that $W_k$ is a nonnegative supermartingale.

\subsection{Proof: Bounding the Exit Probability via Ville's Inequality}

 Since $W_k$ is a nonnegative supermartingale, we can apply Ville's inequality to establish: 
\begin{align}
\label{eq:villeB}
    \P \left\{ \max_{k \in \{ 0, \dots, K\} } W_k > \lambda  \right\} \leq \frac{\mathbb{E}[W_0] }{\lambda} = \frac{W_0}{\lambda}. 
\end{align}
for all $\lambda\in\R_{>0}$. 
To relate this bound to the $K$-step exit  probability $P_u$, we seek a value of $\lambda$, denoted $\lambda^*$, such that:
\begin{equation}
 \max_{k \in \{0, \ldots, K\}} W_k \leq \lambda^*. \implies \min_{k \in \{0, \ldots, K\}} h(\mb{x}_k) \geq -\gamma.
\end{equation}
In short, we will choose a value of $\lambda^*$ such that all trajectories of $W_k$ that remain below $\lambda^*$ must also have $h_k \geq -\gamma$. To this end, we use the geometric series identity\footnote{At $\theta =1 $, the fraction $\frac{1 - \theta^k}{1 - \theta}$ is not well defined. However, the proof can be carried out using the summation notation. In this case $\lambda^* = M + \gamma$, and \eqref{eq:villeB} yields $P_u \leq 1 - \frac{h(\mb{x}_0) + \gamma - \varphi K }{M + \gamma}$. } $\sum_{i=1}^k \theta^{i-1} =\frac{1 - \theta^k}{1 - \theta}$ to rewrite $W_k$ as:
\begin{align}
\label{eqn:wkgeomid}
    W_k &= (M - h(\mb{x}_k))\theta^k + \varphi \theta \frac{\theta^K - \theta^k}{\theta-1}.
\end{align}
Let us define:
\begin{align}
    \lambda_k = \left( \gamma + M - \frac{\varphi \theta}{\theta -1}\right)\theta^{k} + \frac{\varphi \theta}{\theta-1} \theta^K > 0, 
\end{align}
which, intuitively, applies the same time-varying scaling and shift to a constant, $-\gamma$, that was applied to $h(\mb{x}_k)$ to yield $W_k$ \eqref{eqn:wkgeomid}. Let us choose:
\begin{align}
    \lambda^* \triangleq \min_{k \in \{0, \ldots, K\}} \lambda_k.
\end{align}
Since we assume $\max_{k \in \{0, \ldots, K\}} W_k \leq \lambda^*,$ we can write, for all $k \in \{0, \ldots, K\}$:
\begin{align}
    0 &\geq W_k - \lambda^* \geq W_k - \lambda_k = (-\gamma - h_k) \theta^k.
\end{align}
Since $\theta > 1,$ this implies that $-\gamma - h_k \leq 0$ for all $k \in \{0, \ldots, K\}$, and thus $\min_{k \in \{0, \ldots, K\}} h(\mb{x}_k) \geq -\gamma,$ as needed.

\subsection{Proof: Choosing $\theta$ to Minimize the Ville's Bound}

Since our supermartingale $W_k$ includes a free parameter $\theta \in (1, \frac{1}{\alpha}]$, we will choose a value of $\theta$ in this interval which provide the tightest bound on $P_u$.

\textbf{Case 1: } Consider the first case where $\delta < -\gamma(1 - \alpha)$, implying $\varphi > (M + \gamma) (1 - \alpha) $. In this case $\frac{1}{\alpha} < \frac{M+\gamma}{M + \gamma - \varphi}$ and thus all of the allowable choices of $\theta \in (1, \frac{1}{\alpha})$ are such that $\theta < \frac{M+ \gamma}{M + \gamma - \varphi}$. Denoting $k^*$ such that $\lambda^* = \lambda_{k^*}$, we have that: 

\begin{align}
    \lambda^* &= \underbrace{\left( \gamma + M - \frac{\varphi \theta}{\theta -1}\right)}_{\leq 0 }\theta^{k^*} +  \frac{\varphi \theta}{\theta -1} \theta^{K}. \label{eq:case1}
\end{align}
\noindent Thus, we know $\min_{k \in \{0, \dots, K\}} \lambda_k $ occurs at $k^* = K$ and so:
\begin{align}
    P_u & \leq \frac{W_0}{\lambda^*} = \frac{M - h(\mb{x}_0)  + \frac{\varphi\theta}{\theta -1}\left( \theta^K - 1\right) }{(M + \gamma)\theta^K}. 
\end{align}
Since this bound is a decreasing function of $\theta$ (as shown in Lemma \ref{lm:decreasing} in Appendix \ref{apdx:kushner_lemmas}), we choose the largest allowable value $\theta^* = \frac{1}{\alpha}$ to achieve the bound: 
\begin{align}
    P_u & \leq \frac{W_0}{\lambda^* } =   \frac{M - h(\mb{x}_0) + \frac{\varphi }{1 - \alpha}\left( \alpha^{-K} -1 \right)  }{ (M + \gamma)\alpha^{-K} }, \\
    & = \left( \frac{M - h(\mb{x}_0) }{M + \gamma }\right) \alpha^K +  \frac{M (1-\alpha) - \delta}{M + \gamma}\sum_{i=1}^K \alpha^{i-1}, 
\end{align}
where we again use the geometric series identity.

\textbf{Case 2: } Now consider the second case where $\delta \geq -\gamma(1 - \alpha)$, so $\varphi \leq (M + \gamma) (1- \alpha) $, which implies that the set $[\frac{M+ \gamma}{M + \gamma - \varphi}, \frac{1}{\alpha}]$ is nonempty. Choosing a value of $\theta$ in this set ensures that: 
\begin{align}
    \lambda^* &= \underbrace{\left( \gamma + M - \frac{\varphi \theta}{\theta -1}\right)\theta^{k^*}}_{\geq 0 } +  \frac{\varphi \theta}{\theta -1} \theta^{K}.
\end{align}
\noindent Thus $\min_{k \in \{0, \dots, K\}} \lambda_k$ occurs at $k^* = 0 $ and:
\begin{align}
    P_u & \leq \frac{W_0}{\lambda } = \frac{(M - h(\mb{x}_0))  + \frac{\varphi\theta}{\theta -1}\left( \theta^K - 1\right) }{(M + \gamma) + \frac{\varphi \theta}{\theta - 1}\left( \theta^K - 1\right)},\\
    & = 1 -   \frac{h(\mb{x}_0) + \gamma }{ M + \gamma + \frac{\varphi \theta }{\theta - 1}\left( \theta^K -1 \right) }. 
\end{align}
Since this bound 
is increasing in $\theta$ (as shown in Lemma \ref{lm:increasing} in Appendix \ref{apdx:kushner_lemmas}), we choose 
$\theta^* = \frac{M+ \gamma}{M + \gamma - \varphi}$ to achieve the bound: 
\begin{align}
    P_u \leq 1 - \left( \frac{h(\mb{x}_0) + \gamma}{M + \gamma}\right)\left( \frac{M\alpha + \gamma + \delta}{M+ \gamma}\right)^K.
\end{align} 

If, alternatively, we choose $\theta \in \left(1,  \frac{M + \gamma }{M + \gamma - \varphi }\right] $, then the inequality in \eqref{eq:case1} holds, $k^* = K$, and the bound is decreasing in $\theta$ as in Case 1. Evaluating this bound for the minimizing value $\theta^* = \frac{M + \gamma}{M + \gamma - \varphi }$ again yields: 
\begin{align}
    P_u &\leq \frac{M - h(\mb{x}_0) + (M + \gamma) ( \theta^K - 1) }{(M + \gamma) \theta^K},\\
    & = 1 - \left(\frac{h(\mb{x}_0) + \gamma}{M + \gamma}\right)\left( \frac{M\alpha + \gamma + \delta}{M + \gamma }\right)^K. 
\end{align}
\hfill $\blacksquare$

\section{Practical Considerations for Enforcing Stochastic DTCBFs}

Theorem \ref{thm:kushner_main} allows us to reason about the finite-time safety of systems governed by DTBFs. To utilize the results of this theorem in a control setting, we aim to use DTCBFs to develop control methods which enforce the expectation condition:
\begin{align}
    \E[h(\mb{F}(\mb{x}_k, \mb{u}_k)+ \mb{d}_k) \mid \mb{x}_k] & \geq  \alpha h(\mb{x}_k). \label{eq:stochastic_dtcbf_constraint} 
\end{align}
Like the \ref{eq:dtcbfop} controller, we seek to enforce this constraint using an optimization-based controller that enforces safety while achieving pointwise minimal deviation from a nominal controller $\mb{k}_\textrm{nom}$ in the form of an \underline{E}xpectation-based \underline{D}TCBF \eqref{eq:dtcbf_op} Controller:
\begin{align}
    \mb{k}_{\textrm{ED}}(\mb{x}_k) = \argmin_{\mb{u} \in \R^m } \quad & \Vert \mb{u} - \mb{k}_{\textrm{nom}}(\mb{x}_k,k)  \Vert^2 \label{eq:dtcbf_op} \tag{ED}\\ 
    \textrm{s.t. } \quad & \E [h(\mb{F}(\mb{x}_k, \mb{u})+ \mb{d}_k) \mid \mb{x}_k  ] \geq \alpha h(\mb{x}_k). \nonumber
\end{align}

The expectation in \eqref{eq:dtcbf_op} adds complexity that is not generally considered in the application of deterministic DTCBFs. More commonly, CBF-based controllers solve ``certainty-equivalent'' optimization programs, like this \underline{C}ertainty-\underline{E}quivalent \underline{D}TCBF \eqref{eq:CE_dtcbf} controller, that replaces the expected barrier value $\E[h(\mathbf{x}_{k+1})\mid\mb{x}_k]$ with the barrier evaluated at the expected next state, $h(\E[\mb{x}_{k+1}\mid\mb{x}_k])$:
\begin{align}
    \mb{k}_\text{CED}(\mb{x}_k) = \argmin_{\mb{u} \in \R^m } \quad & \Vert \mb{u} - \mb{k}_{\textrm{nom}}(\mb{x}_k,k)  \Vert^2 \label{eq:CE_dtcbf} \tag{CED}\\ 
    \textrm{s.t. } \quad & h(\mb{F}(\mb{x}_k, \mb{u})+ \E[\mb{d}_k])  \geq \alpha h(\mb{x}_k). \nonumber
\end{align}
where $\E[\mb{F}(\mb{x}_k,\mb{u}_k)|\mb{x}_k] = \mb{F}(\mb{x}_k,\mb{u}_k)$ and $\E[\mb{d}_k|\mb{x}_k] = \E[\mb{d}_k]$. This constraint is often easier to evaluate than \eqref{eq:stochastic_dtcbf_constraint} since it allows control actions to be selected with respect to the expected disturbance $\E[\mb{d}_k]$ without needing to model the disturbance distribution $\mathcal{D}$. If the disturbance is zero-mean, then this form of the constraint is implicitly enforced by DTCBF controllers such as those presented in \cite{agrawal2017discrete, zeng2021safety}. However, when replacing \ref{eq:dtcbf_op} with \ref{eq:CE_dtcbf} it is important to consider the effect of Jensen's inequality in Theorem \ref{thm:jensen}.

If the ``certainty-equivalent'' constraint in \ref{eq:CE_dtcbf} is strictly concave\footnote{The constraint $h(\mb{x}_k + \mb{u}) \geq \alpha h(\mb{x}_k)$ is concave in $\mb{u}$ when $h$ is convex and it is convex in $\mb{u}$ when $h$ is concave. }, then we can apply the results of Theorem \ref{thm:kushner_main} directly since Jensen's inequality tightens the constraint and ensures satisfaction of the expectation condition \eqref{eq:kushner_constraint}. Unfortunately, using such a controller is a non-convex optimization program which can be impractical to solve. If, instead, the constraint is convex, then \ref{eq:CE_dtcbf} is a convex program, but does not necessarily enforce the expectation condition \eqref{eq:kushner_constraint} in Theorem \eqref{thm:kushner_main} due to the gap introduced by Jensen's inequality.

In order to apply the results of Theorem \ref{thm:kushner_main} to controllers of the form \eqref{eq:CE_dtcbf} with convex constraints, we must first provide a bound on the gap introduced by Jensen's inequality. In particular, for any concave function $h: \R^n \to \R $ and random variable $\mb{d} \sim \mathcal{D}$, we seek to determine a value $\psi \in \R_{\geq 0}$ such that, for all $\mb{x} \in \R^n$ and $\mb{u}\in\R^m$: 
\begin{align}
    \E[h(\mb{F}(\mb{x}, \mb{u}) +\mb{d}) \mid \mb{x}] \geq h(\mb{F}(\mb{x}, \mb{u}) + \E[\mb{d}]) - \psi, \label{eq:jensen_gap}
\end{align}
\noindent thus quantifying the gap introduced by Jensen's inequality. 

A large body of work has studied methods for finding the smallest possible $\psi$ that satisfies \eqref{eq:jensen_gap}. Here we adapt a result in \cite{becker2012variance} to achieve a relatively loose, but straightforward bound:

\begin{lemma}\label{lm:jensen_gap}
Consider a twice-continuously differentiable, concave function $h: \R^n\to \R$ with $\sup_{\mb{x} \in \R^n} \Vert \nabla^2 h(\mb{x})\Vert_2 \leq \lambda_{\max} $ for some $\lambda_{\max}\in\R_{\geq0}$, and a random variable $\mb{x}$ that takes values in $\R^n$ with $\E[\Vert \mb{x} \Vert] < \infty$ and $\Vert \textup{cov}(\mb{x}) \Vert < \infty$. Then we have that: 
    \begin{align}
        \E[h(\mb{x})] \geq h(\E[\mb{x}]) - \frac{\lambda_{\max}}{2} \textup{tr}(\textup{cov}(\mb{x})).
    \end{align}
\end{lemma}
\noindent The proof is included in Appendix \ref{pf:jensen_gap}. We note that although this value of $\psi= \frac{\lambda_{\textup{max}}}{2}\textrm{tr}(\textrm{cov}(\mb{x}))$ is easy to interpret, tighter bounds exist which have less restrictive assumptions than a globally bounded Hessian \cite{liao2018sharpening}. We also note that one could also use sampling-based methods to approximately satisfy the constraint \eqref{eq:jensen_gap} by estimating $\psi$ empirically.

Next we present a controller which combines the mean-based control of the ``certainty equivalent'' \eqref{eq:CE_dtcbf} while also accounting for Jensen's inequality. This \underline{J}ensen-\underline{E}nhanced \underline{D}TCBF Controller (JED) includes an additional control parameter $c_\textup{J} \geq 0 $ to account for Jensen's inequality: 
\begin{align}
    \mb{k}_\textup{ED}(\mb{x}_k) = \argmin_{\mb{u} \in \R^m } \quad & \Vert \mb{u} - \mb{k}_{\textrm{nom}}(\mb{x}_k,k)  \Vert^2 \label{eq:jed} \tag{JED}\\ 
    \textrm{s.t. } \quad & h(\mb{F}(\mb{x}_k, \mb{u}_k)+ \E[\mb{d}_k])  - c_\textup{J} \geq \alpha h(\mb{x}_k). \nonumber
\end{align}

Given this controller and a method for bounding $\psi$, we can now apply Theorem \ref{thm:kushner_main} while accounting for (or analyzing) the effects of Jensen's inequality on the \eqref{eq:jed} controller: 

\begin{theorem}\label{thm:kushner_jensen}

Consider the system \eqref{eq:dt_autonomous_dist} and let $h:\R^n \to \R$ be a twice-continuously differentiable, concave function such that $\sup_{\mb{x} \in \R^n} h(\mb{x}) \leq M$ for $M\in\R_{>0}$ and  $\sup_{\mb{x} \in \R^n} \Vert \nabla^2 h(\mb{x}) \Vert_2 \leq \lambda_{\max} $ for $\lambda_{\max}\in\R_{\geq0}$. Suppose there exists an $\alpha \in (0,1)$ and a $c_\textup{J} \in [0, \frac{\lambda_\textup{max}}{2}\textup{tr(cov}(\mb{d}))+ M(1-\alpha) ]$ such that: 
\begin{align}
    h(\mb{F}(\mb{x}, \mb{k}(\mb{x})) + \mathbb{E}[\mb{d}] ) - c_\textup{J} \geq \alpha h(\mb{x}), \label{eq:jensen_dtcbf} 
\end{align}
for all $\mb{x} \in \R^n$ with $\mb{d}\sim\mathcal{D}$. Then we have that:
\begin{equation}
\mathbb{E}[~h(\mb{F}(\mb{x}, \mb{k}(\mb{x})) + \mb{d})  \mid \mb{x}~] \geq \alpha h(\mb{x})+ \delta,
\end{equation}
for all $\mb{x}\in\R^n$ with $\mb{d}\sim\mathcal{D}$ and $\delta = c_\textup{J} - \frac{\lambda_\textup{max}}{2}\textup{tr(cov}(\mb{d}_k))$. 
   
\end{theorem}

\begin{proof}
    Given $\mb{x} \in \R^n $, Lemma \ref{lm:jensen_gap} ensures that:
    \begin{align}
        0 &\leq h(\mb{F}(\mb{x}, \mb{k}(\mb{x})) + \mathbb{E}[\mb{d}] )  - c_\textup{J} - \alpha h(\mb{x})\\
        & \leq  \E[h(\mb{F}(\mb{x}, \mb{k}(\mb{x})) + \mb{d} )\mid\mb{x}] +  \psi - c_\textup{J} -\alpha h(\mb{x})  
    \end{align}
    where $\psi = \frac{\lambda_\textup{max}}{2}\textup{tr(cov}(\mb{d})) $. Letting $\delta = c_\textup{J} -\frac{\lambda_\textup{max}}{2}\textup{tr(cov}(\mb{d}))$ yields the the desired result.  
\end{proof}

\section{Practical Examples}
\begin{figure}[t]
    \centering
    \includegraphics[width=\linewidth]{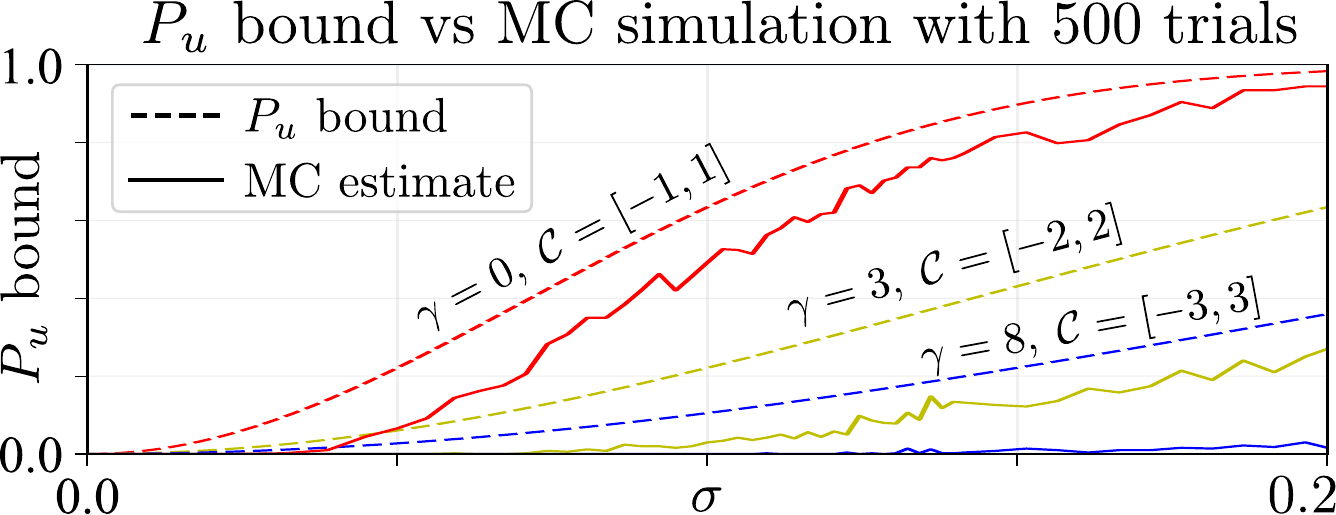}
    \caption{ The dashed lines represent the theoretical probability bounds for the system as in Theorem \ref{thm:kushner_main}. The solid lines represent the Monte Carlo (MC) estimated $P_u$ across 500 experiments.  }   
    \label{fig:steinhardt_comparison}
    \vspace{-1cm}
\end{figure}

In this section we consider a variety of simulation examples that highlight the key features of our approach.

\subsection{Linear 1D System}

Here we analyze our bounds by considering the case of unbounded i.i.d. disturbances $d_k \sim \mathcal{N}(0,1)$ for the one dimensional system ($x, u, \in \R$) and safe set: 
\begin{align}
    x_{k+1} = x_k + 2 + u_k + \sigma d_k \textrm{,  } \;\mathcal{C} = \{ x  \mid 1-x^2 \geq 0 \}.
\end{align}
   
The Jensen gap for this system and DTCBF is bounded by $\psi= \sigma^2$. For simulation, we employ the \ref{eq:jed} controller with $c_\textup{J}= \sigma^2$, $\alpha =1 - \sigma^2$, and nominal controller $\mb{k}_{\textrm{nom}}(\mb{x}_k, k) = 0 $. Figure \ref{fig:steinhardt_comparison} shows the results of 500 one second long trials run with a variety of  $\sigma \in [0, 0.2] $ and also displays how the bound on $P_u$ decreases as $\gamma$ increases.

\subsection{Simple Pendulum}
Next we consider an inverted pendulum about its upright equilibrium point with the DT dynamics: 
\begin{align}
    \lmat 
    \theta_{k+1}\\ 
    \dot{\theta}_{k+1} 
    \rmat 
    = 
    \lmat 
        \theta_k + \Delta t \dot{\theta}_k\\
        \dot{\theta}_k + \Delta t \sin(\theta_k) 
    \rmat 
    + 
    \lmat 
    0\\
    \Delta t \mb{u}
    \rmat 
    +  \mb{d}_k,
\end{align}
\noindent with time step $\Delta_t = 0.01 $ sec, i.i.d disturbances $\mb{d}_k \sim \mathcal{N}(\mb{0}_2,\textrm{Diag}([0.005^2, 0.025^2]]) $, and safe set\footnote{Diag$: \R^n \to \R^{n\times n}$ generates a square diagonal matrix with its argument along the main diagonal.}: 
\begin{align}
    \mathcal{C} = \bigg \{ \mb{x} \in \R^n ~\bigg|~ \underbrace{1 - \frac{6^2}{\pi^2} \mb{x}^\top \lmat 1 & 3^{-\frac{1}{2}} \\ 3^{-\frac{1}{2}} & 1    \rmat \mb{x}}_{h_\textrm{pend}(\mb{x})}  \geq 0  \bigg\}
\end{align}
\noindent which is constructed using the continuous-time Lyapunov equation as in \cite{taylor_safety_2022} and for which  $\vert \theta \vert  \leq \pi/6$ for all $\mb{x} \in \mathcal{C}$. Figure \ref{fig:pendulum} shows the results of 500 one second long trials for each $\mb{x}_0 \in \mathcal{C}$ using the \ref{eq:jed} controller with parameters $\alpha = 1 - \psi, c_\textup{J} = \psi $, where $\psi = \frac{\lambda_\textrm{max}}{2}\textrm{tr(cov(} \mb{d}_k)) $.  
This figure highlights the influence of $\mb{x}_0$ and shows how the bound on $P_u$ increases as $h(\mb{x}_0)$ decreases. 

\begin{figure}[t]
    \centering
    \includegraphics[width=\linewidth]{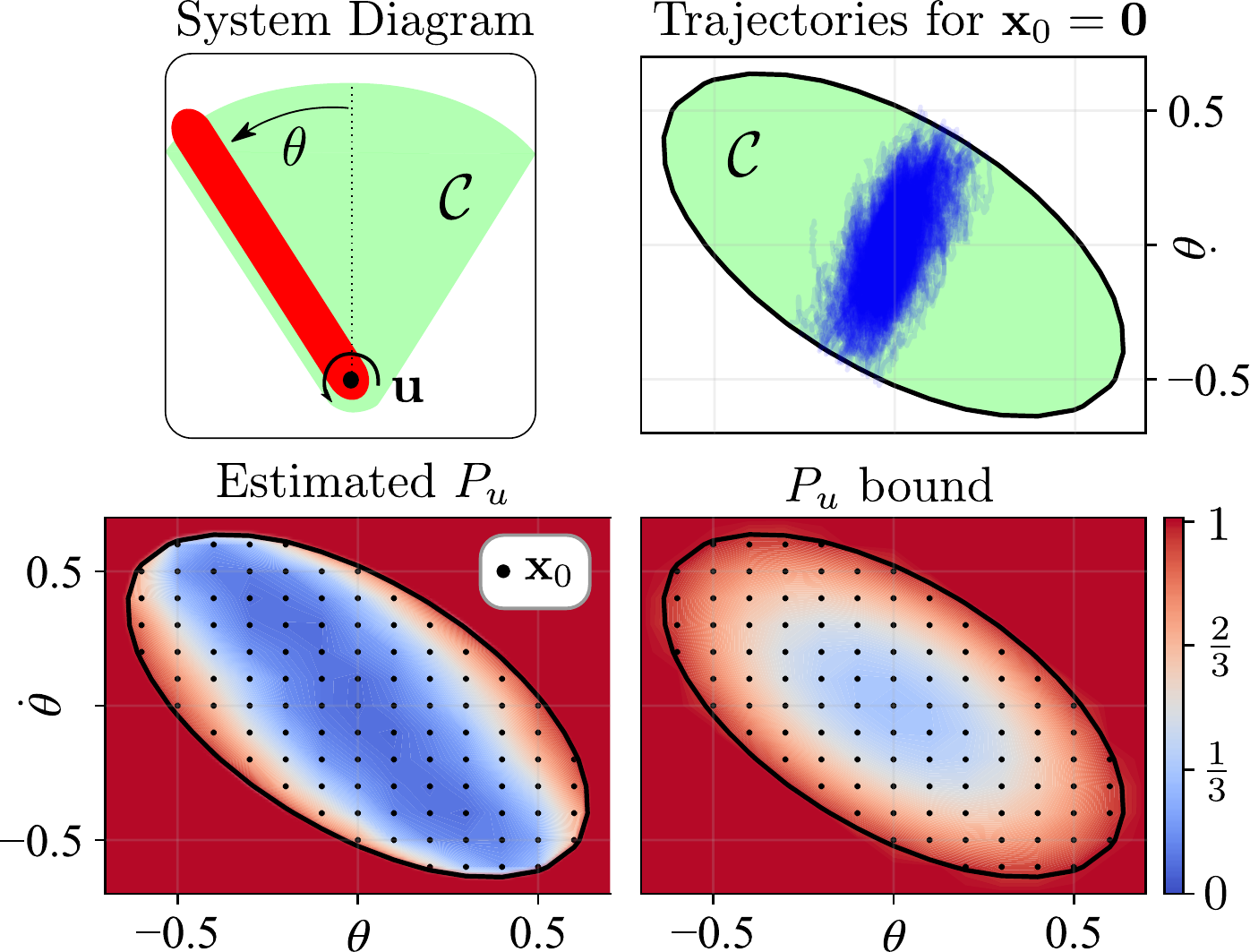}
    \caption{\textbf{(Top Left) }System diagram of the inverted pendulum. \textbf{(Top Right)} 500 one second long example trajectories starting at $\mb{x}_0 = 0 $. \textbf{(Bottom Left) } Monte Carlo estimates of $P_u$ for $\gamma = 0 $ using 500 one second long trials for each initial conditions represented by a black dot. \textbf{(Bottom Right) }  Our (conservative) theoretical bounds on $P_u$ from Theorem \ref{thm:kushner_main}    }
    \label{fig:pendulum}
    \vspace{-1cm}
\end{figure}

\subsection{Double Integrator}
\begin{figure}[t]
    \centering
    \includegraphics[width=0.97\linewidth]{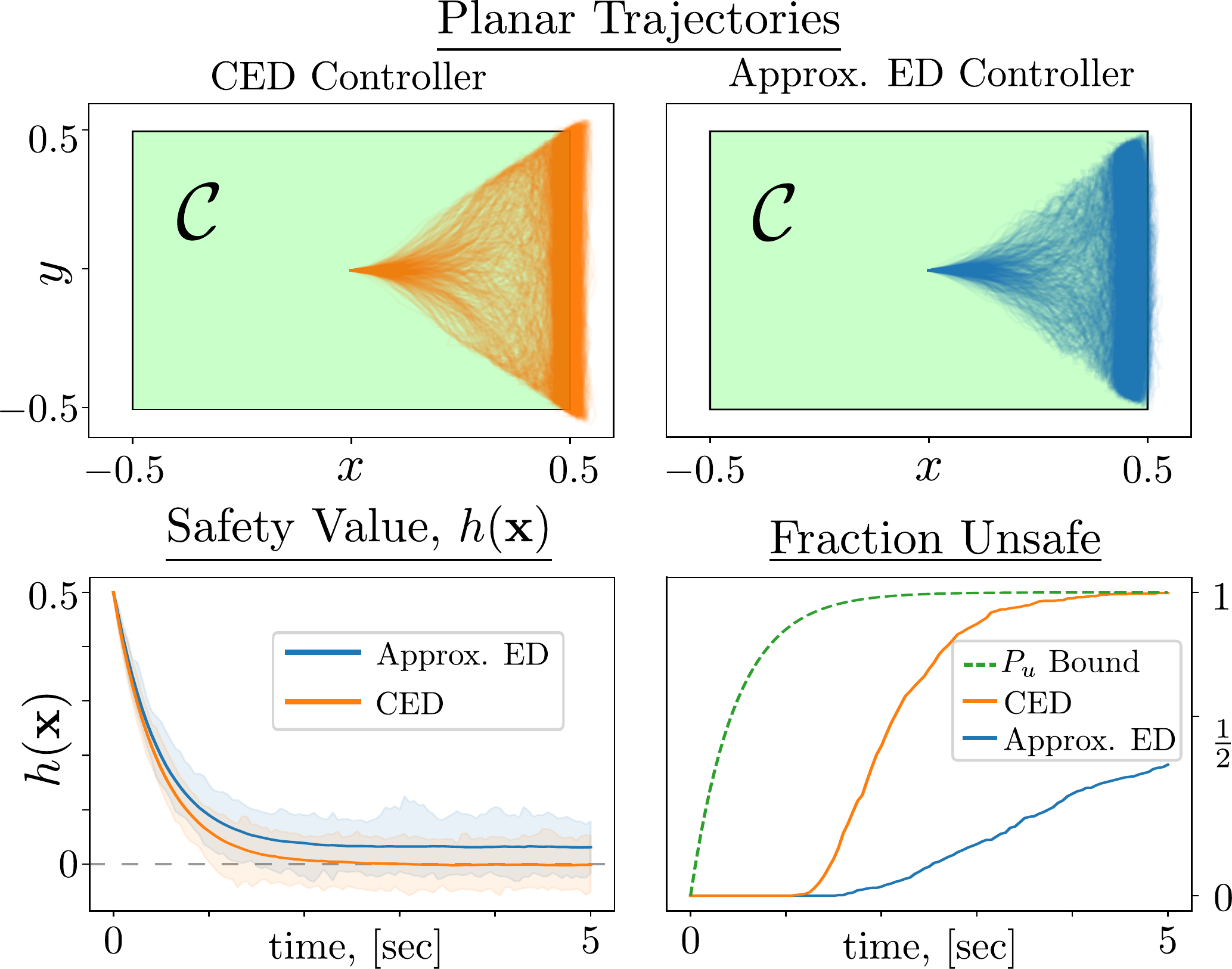}
    \caption{Simulation results for double integrator over $500$ trials. \textbf{(Top left):} Planar ($x,y$) trajectories for the approximated \ref{eq:dtcbf_op} controller, with the safe set (a unit square) plotted in green. \textbf{(Top right):} Planar ($x,y$) trajectories for a \ref{eq:CE_dtcbf} controller. \textbf{(Bottom left):} The $h(\mb{x}_k)$ for both controllers, with the max and min values shaded. \textbf{(Bottom right):} Percent of trajectories that have remained safe over time. We also plot our (conservative) bound \eqref{eq:problo} on the unsafe probability $P_u$.}
    \label{fig:double-integrator}
    \vspace{-0.6cm}
\end{figure}
We also consider the problem of controlling a planar system with unit-mass double-integrator dynamics to remain inside a convex polytope (in particular, a unit square centered at the origin). Using Heun's method, the 
dynamics are given by:
\begin{align}
    \mb{x}_{k+1} &= \left[\begin{array}{cc}\mb{I}_2 & \Delta t \; \mb{I}_2 \\ \mb{0}_2 & \mb{I}_2 \end{array}\right] \mb{x}_k + \left[\begin{array}{c}\frac{\Delta t^2}{2}\mb{I}_2\\  \Delta t \mb{I}_2\end{array}\right] \mb{u}_k + \mb{d}_k,\\
    &\triangleq \mb{A} \mb{x}_k + \mb{B} \mb{u}_k + \mb{d}_k, \label{eq:linear-gaussian}
\end{align}
where $\Delta t$ is the integration time step and $\mb{d}_k \sim \mathcal{N}(\mb{0}_4, \mb{Q})$ is a zero-mean Gaussian process noise added to the dynamics. 
Here we use $\Delta t = 0.05$ sec, and $\mb{Q} = \mb{B} \mb{B}^T$, which corresponds to applying a disturbance force $\mb{f}_k \sim \mathcal{N}(0, \mb{I}_2)$ to the system at each timestep. 

To keep the system inside a convex polytope, we seek to enforce the affine inequalities $\mb{C} \mb{x} \leq \mb{w}$ for $\mb{C} \in \mathbb{R}^{n_c \times n}, \mb{w} \in \mathbb{R}^{n_c}.$ Thus, we define our barrier $h(\mb{x}) = -\max(\mb{C}\mb{x} - \mb{w})$, where $\max(\cdot)$ defines the element-wise maximum, and $h(\mb{x}) \geq 0$ if and only if the constraint $\mb{C} \mb{x} \leq \mb{w}$ holds. Implementing the 
\ref{eq:dtcbf_op} controller for this system is non-trivial, since the expectation of $h(\mathbf{x})$ for a Gaussian-distributed $\mb{x}$ does not have a closed form. Similarly, implementing the \ref{eq:jed} controller to account for Jensen's inequality is non-trivial since $h$ is not twice continuously differentiable. We instead choose to enforce a conservative approximation of the barrier condition \eqref{eq:stochastic_dtcbf_constraint} using the \textit{log-sum-exp} function. As we show in Appendix \ref{apdx:polytope}, this approximation yields an analytic upper bound (derived using the moment-generating function of Gaussian r.v.s) on $\E[h(\mb{x}_{k+1})]$ which can be imposed via a convex constraint.

Figure \ref{fig:double-integrator} plots the results of 500 simulated trajectories for the double integrator system using the proposed \ref{eq:dtcbf_op} controller, and the certainty equivalent \ref{eq:CE_dtcbf} controller that neglects the presence of process noise. Both controllers have a nominal controller $\mb{k}_\text{nom}(\mb{x}) = \left[50, 0\right]$ which seeks to drive the system into the right wall. All trajectories start from the origin. We note the proposed controller is indeed more conservative than the \ref{eq:CE_dtcbf} controller, yielding both fewer and smaller violations of the safe set. In the bottom right, we also plot our bound as a function of the time horizon, which we note is quite conservative compared to our Monte Carlo estimate of the safety probability, motivating future work.

\subsection{Quadruped}
Finally, we consider the problem of controlling a simulated quadrupedal robot locomoting along a narrow path. The simulation is based on a Unitree A1 robot as shown in Figure \eqref{fig:hero_fig} which has 18 degrees of freedom and 12 actuators. An ID-QP controller designed using concepts in \cite{buchli2009inverse} and implemented at 1kHz is used to track stable walking gaits with variable planar velocities and angle rate using the motion primitive framework presented in \cite{Ubellacker2021}. We simulate the entire quadruped's dynamics at 1kHz, but follow a similar methodology to \cite{molnar_model-free_2022} and consider the following simplified discrete-time single-integrator system for DTCBF-based control: 
\begin{align}
    \mb{x}_{k+1}
    =  \mb{x}_k
    + \Delta t\lmat \cos\theta & - \sin \theta  & 0 \\ \sin \theta & \cos\theta & 0 \\ 0 & 0 & 1\rmat \lmat v^x_k \\ v^y_k \\ \theta_k \rmat + \mb{d}_k. \label{eq:reduced_order_quad}  
\end{align}
where $\mb{x}_k = \lmat x, & y, & \theta \rmat^\top $. In order to represent the error caused by uncertain terrain, zero mean Gaussian disturbances are added to the quadruped's $(x,y)$ body position and velocity with variances of $2.25\times 10^6$ and $0.01$ respectively. This random noise along with the dynamics-mismatch between the full-order quadrupedal dynamics and \eqref{eq:reduced_order_quad} is modeled as an i.i.d. random process $\mb{d}_k$. 

The quadruped is commanded to stand and then traverse a 7 meter path that is 1 meter wide, with the safe set
$ \mathcal{C} = \{ \mathbf{x} \in \R^n \mid 0.5^2 - y^2  \geq 0 \} $. For this simulation, three controllers are compared: a simple nominal controller $\mb{k}_\textrm{nom}(\mb{x})  = \lmat 0.2, & 0, & -\theta \rmat^\top $ with no understanding of safety, the \ref{eq:dtcbfop} controller with $\alpha = 0.99$, and our proposed \ref{eq:jed} controller with $\alpha =0.99$ and $c_\textup{J} = \psi$ using the mean and covariance estimates, $\E[\mb{d}_k] \approx  \lmat -0.0132, & -0.0034, & -0.0002 \rmat^\top$ and $\textrm{tr(cov}(\mb{d}_k)) \approx \psi = 0.000548$, which were generated using 15 minutes of 
walking data controlled by $\mb{k}_\textrm{nom}$.

The results of 50 trials for each controller can be seen in Figure \ref{fig:hero_fig}. As expected, $\mb{k}_\textrm{nom}$ generated the largest safety violations and \ref{eq:jed} the smallest and fewest safety violations.

\section{Conclusion} 
\label{sec:conclusion}
In this work, we developed a bound for the finite-time safety of stochastic discrete-time systems using discrete-time control barrier functions. Additionally, we presented a method for practically implementing convex optimization-based controllers which satisfy this bound by accounting for or analyzing the effect of Jensen's inequality. We presented several examples which demonstrate the efficacy of our bound and our proposed \ref{eq:dtcbf_op} and \ref{eq:jed} controllers,

This paper offers a large variety of directions for future work. In particular, in our practical examples, we find the safety bound presented here is often quite conservative in practice. One way forward would be to find other supermartingale transformations of the process $h(\mb{x}_k)$ (perhaps programatically, as in \cite{steinhardt2012finite}) that can yield tighter bounds than those in Theorem \ref{thm:kushner_main}. Another potential avenue may consider alternative martingale inequalities to the Ville's inequality used in this work. Another important open question is how to incorporate state uncertainty into our framework. This would allow us to reason about the safety of CBF-based controllers that operate in tandem with state estimators such as Kalman Filters or SLAM pipelines. Similarly, our methods may have interesting applications in handling the dynamics errors introduced in sampled-data control which can perhaps be modeled as a random variable or learned using a distribution-generating framework such as a state-dependent Gaussian processes or Bayesian neural networks. Finally, we assume that the disturbance distribution $\mathcal{D}$ is known exactly, \textit{a priori}; it would be interesting to consider a ``distributionally robust'' variant of the stochastic barrier condition \eqref{eq:stochastic_dtcbf_constraint} that can provide safety guarantees for a class of disturbance distributions.


\section*{Acknowledgments}
The authors would like to thank Alexander De Capone and Victor Dorobantu 
for their invaluable discussion and Joel Tropp for his course on Probability. The authors would also like to thank Wyatt Ubellacker for generously providing his fantastic quadruped simulation environment. 


\appendix

\subsection{Lemmas for Theorem \ref{thm:kushner_main}} \label{apdx:kushner_lemmas}
The following lemmas are used to prove optimality of the bound in Theorem \ref{thm:kushner_main} Cases 1 and 2. These lemmas were originally stated without proof in \cite{kushner1967stochastic}. 

\begin{lemma}\label{lm:decreasing}
    For $M\in\R_{>0}$, $ \gamma ,\varphi \in\R_{\geq 0}$ , $h(\mb{x}_0) \in [-\gamma, M]$, and $K \in \mathbb{N}_{\geq1}$, the function $\Psi_1:(1,\infty)\to\R$ defined as: 
    \begin{align}
        \Psi_1(\theta) = \frac{M - h(\mb{x}_0) + \frac{\varphi\theta }{\theta -1} \left( \theta^K -1 \right)  }{ (M + \gamma)\theta^K },  
    \end{align}
   is monotonically decreasing. 
\end{lemma}

\begin{proof}
The geometric series identity yields: 
\begin{align}
    \Psi_1(\theta) 
    & = \frac{M - h(\mb{x}_0) }{M + \gamma} \theta^{-K} + \frac{\varphi}{(M+ \gamma) }\sum_{i=1}^K \theta^{i-K},\\
    \frac{d \Psi_1}{d\theta} & =  - \frac{M - h(\mb{x}_0) }{M + \gamma }K \theta^{-K-1} - \varphi\sum_{i=1}^K \frac{(K -i )\theta^{i-K-1}}{M + \gamma}, \nonumber\\
    & \leq 0,
\end{align}
for all $\theta\in(1,\infty)$.
\end{proof}

\begin{lemma}\label{lm:increasing}
For $M\in\R_{>0}, \gamma,\varphi \in\R_{\geq 0} $, $h(\mb{x}_0) \in [-\gamma, M]$, and $K \in \mathbb{N}_{\geq 1} $, the function $\Psi_2:(1,\infty)\to\R$ defined as:  
    \begin{align}
        \Psi_2(\theta) & = 1 -   \frac{h(\mb{x}_0) + \gamma }{ M + \gamma + \frac{\varphi \theta }{\theta - 1}\left( \theta^K -1 \right) },  
    \end{align}
    is monotonically increasing. 
\end{lemma}

\begin{proof} The geometric series identity yields:
\begin{align}
    \Psi_2(\theta) 
    & = 1 -   \frac{h(\mb{x}_0) + \gamma }{ M + \gamma + \varphi\sum_{i=1}^K \theta^i  },  \\
    \frac{d\Psi_2}{d\theta} & = \frac{(h(\mb{x}_0) + \gamma) \left(\varphi \sum_{i=1}^K i \theta^{i-1}\right)}{\left( M + \gamma + \varphi \sum_{i=1}^K \theta^{i}\right)^2 },\\
    & \geq 0, 
\end{align}
for all $\theta\in(1,\infty)$.
\end{proof}

\subsection{Lemma \ref{lm:jensen_gap}}
Here we present a proof of Lemma \ref{lm:jensen_gap}. 

\begin{proof} \label{pf:jensen_gap}
Consider the convex, twice-continuously differentiable function $\eta:\R^n\to\R$ defined as $\eta = -h$. The intermediate value theorem implies that for all $\mb{y},\mb{z}\in\R^n$, there exists an $\omega \in [0,1]$ such that: 
\begin{align}
    \eta(\mb{z}) & = \eta(\mb{y}) + \nabla \eta(\mb{y})^\top \mb{e} + \frac{1}{2}\mb{e}^\top\nabla^2\eta(\mb{c}) \mb{e}, 
\end{align}
where $\mb{e}\triangleq \mb{z} - \mb{y} $,  $  \mb{c}  \triangleq \omega \mb{z} + (1- \omega) \mb{z}$, and $\nabla^2 \eta(\mb{c}) $ is the Hessian of $\eta$ evaluated 
at $\mb{c}$. We then have that: 
\begin{align}
            \eta(\mb{z})
            & = \eta(\mb{y}) + \nabla \eta(\mb{y})^\top \mb{e}  + \frac{1}{2}\textrm{tr}\left(\nabla^2\eta(\mb{c}) \mb{e}\mb{e}^\top \right),\\
            & \leq \eta(\mb{y}) + \nabla \eta(\mb{y})^\top \mb{e}  + \frac{1}{2}\Vert \nabla^2\eta(\mb{c})\Vert_2\textrm{tr}\left(\mb{e}\mb{e}^\top \right),\\
            & \leq  \eta(\mb{y}) + \nabla \eta(\mb{y})^\top \mb{e} + \frac{\lambda_\textrm{max} }{2}\textrm{tr}\left( \mb{e}\mb{e}^\top \right),
\end{align}
where the first inequality is a property of the trace operator for positive semi-definite matrices \cite{shebrawi2013trace} (and $\nabla^2\eta(\mb{c})$ is positive semi-definite as $\eta$ is convex), and the second inequality follows by our definition of $\lambda_{\rm max}$. Let $\mb{x}$ be a random variable taking values in $\R^n$ with probability density function $p: \R^n \to \R_{\geq 0}$, and let $\bs{\mu} \triangleq \E[\mb{x}]$. We then have that:
\begin{align}
    & \E[\eta(\mb{x})] - \eta(\E[\mb{x}]) = \int_{\R^n} (\eta(\mb{x}) - \eta(\bs{\mu}) )p(\mb{x}) d\mb{x}, \\
    & \leq \int_{\R^n }   \nabla \eta(\bs{\mu})^\top \mb{e} + \frac{\lambda_\textrm{max} }{2}\textrm{tr}\left( \mb{e}\mb{e}^\top \right)p(\mb{x}) d\mb{x},\\
    & = \frac{\lambda_\textrm{max} }{2} \textrm{tr}(\textup{cov}(\mb{x}) ), 
\end{align}
where $\mb{e} = \mb{x}-\bs{\mu}$. Replacing $\eta$ with $-h$ yields: 
\begin{align}
    \E[h(\mb{x})] \geq  h(\E[\mb{x}]) - \frac{\lambda_\textrm{max}}{2}\textrm{tr}(\textup{cov}(\mb{x}) ). 
\end{align}

\end{proof}

\subsection{Derivation of Convex Approximation for Polytopic Barrier}
\label{apdx:polytope}

Here we derive a conservative approximation of the constraint $\E[h(\mb{x}_{k+1})] \geq \alpha h(\mb{x}_k)$ for barriers of the form $h(\mb{x}) = -\max(\mb{C} \mb{x} - \mb{w})$ and systems with linear-Gaussian dynamics \eqref{eq:linear-gaussian}. The key idea is to use the \textit{log-sum-exp} function as a smooth, convex upper bound of the pointwise maximum in the barrier function, which yields a closed-form expression for Gaussian random variables.

In particular, if $L$ is the \textit{log-sum-exp} function, for any $t > 0$, $\max(\mb{x}) \leq \frac{1}{t} L(t\mb{x}) \triangleq \frac{1}{t} \log(\sum_{i=1}^n \exp(t x_i))$ \cite[Chapter~3]{boyd2014convex}. We can use this to upper bound the expectation of $-h$,
\begin{align}
    -\E[h(\mb{x}_{k+1})] &= \E\left[\max(\mb{C}\mb{x}_{k+1} - \mb{w})\right] \\ &\leq \frac{1}{t} \E\left[ L\Big(t(\mb{C} \mb{x}_{k+1} - \mb{w})\Big)\right] \\ &\leq \frac{1}{t}\log\left(\sum_{i=1}^{n_c}\E\left[\exp(t\mb{r}_i)\right]\right),
\end{align}
for $\mb{r}_i \triangleq \mb{c}_i^T \mb{x} - w_i$, where $\mb{c}_i$ is the $i^\text{th}$ row of $\mb{C}$, $w_i$ is the $i^\text{th}$ entry of $\mb{w}$, and the last inequality follows from Jensen's inequality and the concavity of the natural logarithm. Further, since we have linear-Gaussian dynamics, it is easy to show that $\mb{r}_i \sim \mathcal{N}(\mb{c}_i^T (\mb{A} \mb{x}_k + \mb{B} \mb{u}_k) - w_i, \mb{c}_i^T \mb{Q} \mb{c}_i)$. The expression $\E[\exp(t \mb{X})]$ is the ``moment-generating function'' of a random variable $\mb{X},$ and for a Gaussian r.v. $\mb{X} \sim \mathcal{N}(\mu, \sigma^2)$, it has a closed form, $\E[\exp(t \mb{X})] = \exp(t \mu + \frac{t^2}{2} \sigma^2)$ \cite[Chapter~6]{wackerly2014mathematical}. 

Thus, for $\bm{\mu} \triangleq \mb{C}(\mb{A} \mb{x}_{k} + \mb{B} \mb{u}_k) - \mb{w},$ $\bm{\sigma} \triangleq \diag(\mb{A} \mb{Q} \mb{A}^T)$, where $\diag(\cdot)$ defines the diagonal of a square matrix, 
\begin{align}
    -\E[h(\mb{x}_{k+1})] &\leq \frac{1}{t}L\left(t \bm{\mu} + \frac{t^2}{2} \bm{\sigma} \right),
\end{align}
which implies that imposing the constraint $\frac{1}{t} L(t \bm{\mu} + \frac{t^2}{2} \bm{\sigma}) \leq -\alpha h(\mb{x}_k)$ ensures that the stochastic barrier condition \eqref{eq:stochastic_dtcbf_constraint} is satisfied. Finally, recognizing that our constraint is a perspective transform of $L(\bm{\mu} + \frac{t}{2} \bm{\sigma})$ by the scalar $\frac{1}{t},$ which preserves convexity \cite[Chapter~3]{boyd2014convex}, our constraint is indeed convex. 
Thus an optimization-based controller such as \ref{eq:dtcbf_op} can be used online to select control actions, and can jointly optimize over $\mb{u}_k, t$ to obtain the tightest bound on the expectation possible.

\bibliographystyle{plainnat}
\bibliography{references, taylor}

\begin{thebibliography}{37}
\providecommand{\natexlab}[1]{#1}
\providecommand{\url}[1]{\texttt{#1}}
\expandafter\ifx\csname urlstyle\endcsname\relax
  \providecommand{\doi}[1]{doi: #1}\else
  \providecommand{\doi}{doi: \begingroup \urlstyle{rm}\Url}\fi

\bibitem[Agrawal and Sreenath(2017)]{agrawal2017discrete}
Ayush Agrawal and Koushil Sreenath.
\newblock Discrete control barrier functions for safety-critical control of
  discrete systems with application to bipedal robot navigation.
\newblock In \emph{Robotics: Science and Systems}, volume~13. Cambridge, MA,
  USA, 2017.

\bibitem[Agrawal et~al.(2022)Agrawal, Parwana, Cosner, Rosolia, Ames, and
  Panagou]{agrawal_constructive_2022}
Devansh~R. Agrawal, Hardik Parwana, Ryan~K. Cosner, Ugo Rosolia, Aaron~D. Ames,
  and Dimitra Panagou.
\newblock A {Constructive} {Method} for {Designing} {Safe} {Multirate}
  {Controllers} for {Differentially}-{Flat} {Systems}.
\newblock \emph{IEEE Control Systems Letters}, 6:\penalty0 2138--2143, 2022.
\newblock ISSN 2475-1456.
\newblock \doi{10.1109/LCSYS.2021.3136465}.
\newblock URL \url{https://ieeexplore.ieee.org/document/9655322/}.

\bibitem[Ahmadi et~al.(2019)Ahmadi, Singletary, Burdick, and
  Ames]{ahmadi2019safe}
Mohamadreza Ahmadi, Andrew Singletary, Joel~W Burdick, and Aaron~D Ames.
\newblock Safe policy synthesis in multi-agent pomdps via discrete-time barrier
  functions.
\newblock In \emph{2019 IEEE 58th Conference on Decision and Control (CDC)},
  pages 4797--4803. IEEE, 2019.

\bibitem[Ahmadi et~al.(2022)Ahmadi, Xiong, and Ames]{ahmadi_risk-averse_2022}
Mohamadreza Ahmadi, Xiaobin Xiong, and Aaron~D. Ames.
\newblock Risk-{Averse} {Control} via {CVaR} {Barrier} {Functions}:
  {Application} to {Bipedal} {Robot} {Locomotion}.
\newblock \emph{IEEE Control Systems Letters}, 6:\penalty0 878--883, 2022.
\newblock ISSN 2475-1456.
\newblock \doi{10.1109/LCSYS.2021.3086854}.
\newblock Conference Name: IEEE Control Systems Letters.

\bibitem[Alan et~al.(2021)Alan, Taylor, He, Orosz, and Ames]{alan2021safe}
Anil Alan, Andrew~J Taylor, Chaozhe~R He, G{\'a}bor Orosz, and Aaron~D Ames.
\newblock Safe controller synthesis with tunable input-to-state safe control
  barrier functions.
\newblock \emph{IEEE Control Systems Letters}, 6:\penalty0 908--913, 2021.

\bibitem[Ames et~al.(2016)Ames, Xu, Grizzle, and Tabuada]{ames2016control}
Aaron~D Ames, Xiangru Xu, Jessy~W Grizzle, and Paulo Tabuada.
\newblock Control barrier function based quadratic programs for safety critical
  systems.
\newblock \emph{IEEE Transactions on Automatic Control}, 62\penalty0
  (8):\penalty0 3861--3876, 2016.

\bibitem[Ames et~al.(2019)Ames, Coogan, Egerstedt, Notomista, Sreenath, and
  Tabuada]{ames_control_2019}
Aaron~D. Ames, Samuel Coogan, Magnus Egerstedt, Gennaro Notomista, Koushil
  Sreenath, and Paulo Tabuada.
\newblock Control {Barrier} {Functions}: {Theory} and {Applications}.
\newblock In \emph{2019 18th {European} {Control} {Conference} ({ECC})}, pages
  3420--3431, June 2019.
\newblock \doi{10.23919/ECC.2019.8796030}.

\bibitem[Aubin et~al.(2011)Aubin, Bayen, and Saint-Pierre]{aubin2011viability}
Jean-Pierre Aubin, Alexandre~M Bayen, and Patrick Saint-Pierre.
\newblock \emph{Viability theory: new directions}.
\newblock Springer Science \& Business Media, 2011.

\bibitem[Bansal et~al.(2017)Bansal, Chen, Herbert, and
  Tomlin]{bansal2017hamilton}
Somil Bansal, Mo~Chen, Sylvia Herbert, and Claire~J Tomlin.
\newblock Hamilton-jacobi reachability: A brief overview and recent advances.
\newblock In \emph{2017 IEEE 56th Annual Conference on Decision and Control
  (CDC)}, pages 2242--2253. IEEE, 2017.

\bibitem[Becker(2012)]{becker2012variance}
Robert~A Becker.
\newblock The variance drain and jensen's inequality.
\newblock 2012.

\bibitem[Blanchini and Miani(2008)]{blanchini2008set}
Franco Blanchini and Stefano Miani.
\newblock \emph{Set-theoretic methods in control}, volume~78.
\newblock Springer, 2008.

\bibitem[Bof et~al.(2018)Bof, Carli, and Schenato]{bof2018lyapunov}
Nicoletta Bof, Ruggero Carli, and Luca Schenato.
\newblock Lyapunov theory for discrete time systems.
\newblock \emph{arXiv preprint arXiv:1809.05289}, 2018.

\bibitem[Boyd and Vandenberghe(2014)]{boyd2014convex}
Stephen~P. Boyd and Lieven Vandenberghe.
\newblock \emph{Convex Optimization}.
\newblock Cambridge University Press, 2014.
\newblock ISBN 978-0-521-83378-3.
\newblock \doi{10.1017/CBO9780511804441}.
\newblock URL \url{https://web.stanford.edu/\%7Eboyd/cvxbook/}.

\bibitem[Buchli et~al.(2009)Buchli, Kalakrishnan, Mistry, Pastor, and
  Schaal]{buchli2009inverse}
Jonas Buchli, Mrinal Kalakrishnan, Michael Mistry, Peter Pastor, and Stefan
  Schaal.
\newblock Compliant quadruped locomotion over rough terrain.
\newblock In \emph{IEEE/RSJ International Conference on Intelligent Robots and
  Systems}, pages 814--820, 2009.
\newblock \doi{10.1109/IROS.2009.5354681}.

\bibitem[Cheng et~al.(2019)Cheng, Orosz, Murray, and Burdick]{cheng2019end}
Richard Cheng, G{\'a}bor Orosz, Richard~M Murray, and Joel~W Burdick.
\newblock End-to-end safe reinforcement learning through barrier functions for
  safety-critical continuous control tasks.
\newblock In \emph{Proceedings of the AAAI Conference on Artificial
  Intelligence}, volume~33, pages 3387--3395, 2019.

\bibitem[Chern et~al.(2021)Chern, Wang, Iyer, and Nakahira]{chern_safe_2021}
Albert Chern, Xiang Wang, Abhiram Iyer, and Yorie Nakahira.
\newblock Safe {Control} in the {Presence} of {Stochastic} {Uncertainties}.
\newblock In \emph{2021 60th {IEEE} {Conference} on {Decision} and {Control}
  ({CDC})}, pages 6640--6645, December 2021.
\newblock \doi{10.1109/CDC45484.2021.9683542}.
\newblock ISSN: 2576-2370.

\bibitem[Clark(2019)]{clark_control_2019}
Andrew Clark.
\newblock Control {Barrier} {Functions} for {Complete} and {Incomplete}
  {Information} {Stochastic} {Systems}.
\newblock In \emph{2019 {American} {Control} {Conference} ({ACC})}, pages
  2928--2935, July 2019.
\newblock \doi{10.23919/ACC.2019.8814901}.
\newblock ISSN: 2378-5861.

\bibitem[Grimmett and Stirzaker(2020)]{grimmett2020probability}
Geoffrey Grimmett and David Stirzaker.
\newblock \emph{Probability and random processes}.
\newblock Oxford university press, 2020.

\bibitem[Hewing et~al.(2020)Hewing, Wabersich, Menner, and
  Zeilinger]{hewing2020learning}
Lukas Hewing, Kim~P Wabersich, Marcel Menner, and Melanie~N Zeilinger.
\newblock Learning-based model predictive control: Toward safe learning in
  control.
\newblock \emph{Annual Review of Control, Robotics, and Autonomous Systems},
  3:\penalty0 269--296, 2020.

\bibitem[Kolathaya and Ames(2018)]{kolathaya2018input}
Shishir Kolathaya and Aaron~D Ames.
\newblock Input-to-state safety with control barrier functions.
\newblock \emph{IEEE control systems letters}, 3\penalty0 (1):\penalty0
  108--113, 2018.

\bibitem[Kushner(1967)]{kushner1967stochastic}
Harold~J Kushner.
\newblock Stochastic stability and control.
\newblock Technical report, Brown Univ Providence RI, 1967.

\bibitem[Liao and Berg(2018)]{liao2018sharpening}
JG~Liao and Arthur Berg.
\newblock Sharpening jensen's inequality.
\newblock \emph{The American Statistician}, 2018.

\bibitem[Liu et~al.(2022)Liu, Zeng, Sreenath, and Belta]{liu2022iterative}
Shuo Liu, Jun Zeng, Koushil Sreenath, and Calin~A Belta.
\newblock Iterative convex optimization for model predictive control with
  discrete-time high-order control barrier functions.
\newblock \emph{arXiv preprint arXiv:2210.04361}, 2022.

\bibitem[Molnar et~al.(2022)Molnar, Cosner, Singletary, Ubellacker, and
  Ames]{molnar_model-free_2022}
Tamas~G. Molnar, Ryan~K. Cosner, Andrew~W. Singletary, Wyatt Ubellacker, and
  Aaron~D. Ames.
\newblock Model-{Free} {Safety}-{Critical} {Control} for {Robotic} {Systems}.
\newblock \emph{IEEE Robotics and Automation Letters}, 7\penalty0 (2):\penalty0
  944--951, April 2022.
\newblock ISSN 2377-3766, 2377-3774.
\newblock \doi{10.1109/LRA.2021.3135569}.
\newblock URL \url{https://ieeexplore.ieee.org/document/9652122/}.

\bibitem[Papachristodoulou and Prajna(2005)]{papachristodoulou2005tutorial}
Antonis Papachristodoulou and Stephen Prajna.
\newblock A tutorial on sum of squares techniques for systems analysis.
\newblock In \emph{Proceedings of the 2005, American Control Conference,
  2005.}, pages 2686--2700. IEEE, 2005.

\bibitem[Prajna et~al.(2004)Prajna, Jadbabaie, and
  Pappas]{prajna2004stochastic}
Stephen Prajna, Ali Jadbabaie, and George~J Pappas.
\newblock Stochastic safety verification using barrier certificates.
\newblock In \emph{2004 43rd IEEE conference on decision and control (CDC)(IEEE
  Cat. No. 04CH37601)}, volume~1, pages 929--934. IEEE, 2004.

\bibitem[Santoyo et~al.(2019)Santoyo, Dutreix, and
  Coogan]{santoyo_verification_2019}
Cesar Santoyo, Maxence Dutreix, and Samuel Coogan.
\newblock Verification and {Control} for {Finite}-{Time} {Safety} of
  {Stochastic} {Systems} via {Barrier} {Functions}, May 2019.
\newblock URL \url{http://arxiv.org/abs/1905.12077}.
\newblock arXiv:1905.12077 [cs].

\bibitem[Shebrawi and Albadawi(2013)]{shebrawi2013trace}
Khalid Shebrawi and Hussien Albadawi.
\newblock Trace inequalities for matrices.
\newblock \emph{Bulletin of the Australian Mathematical Society}, 87\penalty0
  (1):\penalty0 139--148, 2013.

\bibitem[Sontag(2008)]{sontag2008input}
Eduardo~D Sontag.
\newblock Input to state stability: Basic concepts and results.
\newblock In \emph{Nonlinear and optimal control theory}, pages 163--220.
  Springer, 2008.

\bibitem[Steinhardt and Tedrake(2012)]{steinhardt2012finite}
Jacob Steinhardt and Russ Tedrake.
\newblock Finite-time regional verification of stochastic non-linear systems.
\newblock \emph{The International Journal of Robotics Research}, 31\penalty0
  (7):\penalty0 901--923, 2012.

\bibitem[Taylor et~al.(2022)Taylor, Dorobantu, Cosner, Yue, and
  Ames]{taylor_safety_2022}
Andrew~J. Taylor, Victor~D. Dorobantu, Ryan~K. Cosner, Yisong Yue, and Aaron~D.
  Ames.
\newblock Safety of {Sampled}-{Data} {Systems} with {Control} {Barrier}
  {Functions} via {Approximate} {Discrete} {Time} {Models}, June 2022.
\newblock URL \url{http://arxiv.org/abs/2203.11470}.
\newblock arXiv:2203.11470 [cs, eess].

\bibitem[Ubellacker et~al.(2021)Ubellacker, Csomay-Shanklin, Molnar, and
  Ames]{Ubellacker2021}
Wyatt Ubellacker, Noel Csomay-Shanklin, Tamas~G. Molnar, and Aaron~D. Ames.
\newblock Verifying safe transitions between dynamic motion primitives on
  legged robots.
\newblock \emph{arXiv preprint}, \penalty0 (arXiv:2106.10310), 2021.

\bibitem[Ville(1939)]{ville1939etude}
Jean Ville.
\newblock Etude critique de la notion de collectif.
\newblock \emph{Bull. Amer. Math. Soc}, 45\penalty0 (11):\penalty0 824, 1939.

\bibitem[Wackerly et~al.(2014)Wackerly, Mendenhall, and
  Scheaffer]{wackerly2014mathematical}
D.~Wackerly, W.~Mendenhall, and R.L. Scheaffer.
\newblock \emph{Mathematical {{Statistics}} with {{Applications}}}.
\newblock {Cengage Learning}, 2014.
\newblock ISBN 978-1-111-79878-9.

\bibitem[Wills and Heath(2004)]{wills2004barrier}
Adrian~G Wills and William~P Heath.
\newblock Barrier function based model predictive control.
\newblock \emph{Automatica}, 40\penalty0 (8):\penalty0 1415--1422, 2004.

\bibitem[Zeng et~al.(2021)Zeng, Zhang, and Sreenath]{zeng2021safety}
Jun Zeng, Bike Zhang, and Koushil Sreenath.
\newblock Safety-critical model predictive control with discrete-time control
  barrier function.
\newblock In \emph{2021 American Control Conference (ACC)}, pages 3882--3889.
  IEEE, 2021.

\bibitem[Zhou and Doyle(1998)]{zhou1998essentials}
Kemin Zhou and John~Comstock Doyle.
\newblock \emph{Essentials of robust control}, volume 104.
\newblock Prentice hall Upper Saddle River, NJ, 1998.

\end{thebibliography}

\end{document}